\documentclass[11pt,superscriptaddress]{article} 
\usepackage{amsmath}
\usepackage{amsfonts}
\usepackage{amsthm}
\usepackage{amssymb}
\usepackage{url}
\usepackage{hyperref}
\usepackage{ifthen}
\usepackage{color}
\usepackage{graphicx}

\newboolean{qcircuit}

\setboolean{qcircuit}{true}     

\newcommand{\ket}[1]{\left |  #1 \right \rangle}
\newcommand{\bra}[1]{\left \langle #1 \right |}

\newcommand{\proj}[1]{\ket{#1}\!\bra{#1}}

\newcommand{\ketbra}[2]{\left | #1 \rangle \langle #2 \right |}

\newcommand{\norm}[1]{\left|\left|#1\right|\right|}
\newcommand{\II}{\mathbb{I}}
\newcommand{\Tr}{\operatorname{Tr}}

\newtheorem{definition}{Definition}

\newtheorem{theorem}{Theorem}

\newtheorem{lemma}{Lemma}

\newtheorem{protocol}{Protocol}

\begin{document}

\title{Insider-proof encryption with applications for quantum key distribution}
\author{
\begin{tabular}{cc}
Matthew McKague\footnote{\url{matthew.mckague@otago.ac.nz}} 
&
Lana Sheridan\footnote{\url{cqtlss@nus.edu.sg}} \\ 
Department of Physics &
Centre for Quantum Technologies \\ 
University of Otago &
National University of Singapore 
\end{tabular}
}

\maketitle

\begin{abstract}
It has been pointed out~\cite{BCK12} that current protocols for device independent quantum key distribution can leak key to the adversary when devices are used repeatedly and that this issue has not been addressed.  We introduce the notion of an insider-proof channel. This allows us to propose a means by which devices with memories could be reused from one run of a device independent quantum key distribution protocol to the next while bounding the leakage to Eve, under the assumption that one run of the protocol could be completed securely using devices with memories. 
\end{abstract}

\section{Introduction}

Quantum key distribution protocols allow two distant parties who share some small initial key to grow new shared randomness.  Proofs of security for these protocols make assumptions about the behaviour of the devices that the two parties, Alice and Bob, use.  Device-independent quantum key distribution (DIQKD)~\cite{AGM06,ABGMPS07,devindepfull} is a concept for protocols that makes very few assumptions about Alice and Bob's devices for generating their classical measurement outcomes.  They should be able to certify that they have a secure key from the statistics of their measurement outcomes alone, however, they still need to assume that there are no side channels that can signal from their private laboratories to an eavesdropper, Eve.  DIQKD protocols have the important advantage that even if the measurement devices (or the source of quantum states, the control of which we give to Eve) do not operate as intended, the protocols can still certify whether a generated key is secure.  

Here we are interested in removing even more assumptions about the operation of the measurement devices: we allow that they may have an internal memory which can store arbitrary amounts of quantum or classical information, and that they may have been built by the eavesdropper (subject to the usual signalling restrictions).  In this very untrusting model, the question has been raised~\cite{BCK12} whether the measurement devices can be reused.  In this paper, we consider this question, but do not consider how DIQKD might be accomplished in a single round of a protocol using an adversarial device with memory.

Our main contribution is to describe an encryption scheme which allows Alice and Bob to exchange data which is determined by the devices across a public channel without leaking information from the devices to Eve.  The encryption remains secure even if the devices have complete information about Alice and Bob's shared secret keys (generated in previous rounds of the protocol) and even if the devices have complete control over the message sent.  In the context of DIQKD, this allows Alice and Bob to exchange parameter estimation and error correction data without the devices leaking information about previously generated keys to Eve.  This is accomplished using locally generated randomness (independent of the devices) and hash functions to generate encryption keys.

The layout of this paper is as follows.  In the next section, we describe what an ideal private channel consists of and introduce a new cryptographic concept: the $\epsilon$-insider proof channel.  This channel allows the secure transmission of a message, even if the message is chosen adversarially.  In part~\ref{sec:implch} of this section we give a recipe for implementing this channel and in part~\ref{sec:secch} demonstrate the channel does allow secure message transmission except with probability $\epsilon$.  

In section~\ref{sec:model}, we begin to look at applying the $\epsilon$-insider proof channel to quantum key distribution and we specify and motivate the security model we are working in.  Following that, we outline the modifications to a DIQKD protocol in section~\ref{sec:protocol}, and in section~\ref{sec:qkdsec} it is shown that the DIQKD protocol is still secure with the modifications that use the new channel.  The composition of repeated DIQKD rounds is considered in section~\ref{sec:bounds}.  Section~\ref{sec:accounting} gives the asymptotic key rate achieved by these bounds.  Lastly, in section~\ref{sec:aborts} we discuss how protocol aborts need to be managed, touch on the composability implications, and present some conclusions.


\section{The private channel}
\label{sec:channel}
We first describe the scenario and define the ideal insider-proof private channel.  We then give a protocol and prove that it approximates the ideal channel.  The proof relies on 2-universal hashing and the quantum leftover hashing lemma~\cite{TSSR11}.

\subsection{The ideal channel}
Let us define a situation where Alice wishes to privately communicate some information to Bob in the presence of a quantum eavesdropper, Eve, who wishes to obtain access to some of Alice's data.  Further, there is an insider $A^{\prime}$ who has access to Alice's private keys and data, and who can choose some messages to be sent to Bob.  That is, there will be some encrypted channel from Alice to Bob and $A^{\prime}$ can choose some inputs to the channel.  However, $A^{\prime}$ has no other means of communicating with Eve.  Alice and Bob's task is to complete their private communication in such a way that Eve cannot gain any information about Alice's data or the message.

\subsubsection{Definitions}
To begin, let us describe the registers we will use.  $A$ contains the message Alice will send, while $B$ is the register which will hold the final message for Bob.  $C$ is the ciphertext, or otherwise contains all the raw information leaked to Eve during the protocol.  For example, if the channel is implemented using a public channel then $C$ contains all information sent over the channel. $D$ contains secret information that Alice does not want to leak.  Finally, $E$ contains Eve's quantum side information.  We assume that the length of $A$, $B$ and $C$ are public.

\begin{definition}
The \emph{ideal private channel} between $A$ and $B$ is defined as
\begin{equation}
\Phi^{IPC}_{AB}(\rho) = \sum_{x,y} \left(\proj{x}_{A} \otimes \ketbra{x}{y}_{B} \right) \rho
\left( 
\proj{x}_{A} \otimes \ketbra{y}{x}_{B} 
\right)
\end{equation}
\end{definition}
So, the ideal private channel erases $B$ and copies the contents of $A$ into it.

\begin{definition}\label{def:ippc}
A channel $\Phi_{ABC}$ is an \emph{$\epsilon$-insider-proof private channel} from $A$ to $B$ if there exists a channel $\Psi_{CE}$ such that for all $CCCCQ$ states $\rho_{ABCDE}$, 
\begin{equation}
\norm{\Phi_{AB}^{IPC} \otimes I_{D} \otimes \Psi_{CE}(\rho) - 
\Phi_{ABC} \otimes I_{DE}(\rho)}_{1} \leq \epsilon
\end{equation}
Furthermore, if this is true for $\epsilon = 0$ then we say that $\Phi$ is an \emph{insider-proof private channel}.

\end{definition}
This definition essentially says that we consider a channel secure if we can approximate it with the ideal channel, along with some simulator that generates a transcript for Eve without referring to the secret data.

Although we have not explicitly stated that the insider can choose the message, this is built into the fact that we allow \emph{any} $\rho$, and hence this covers the cases where the insider has deliberately correlated the registers $A$, $D$ and its memory $A^{\prime}$, possibly using some quantum measurement on half of an entangled $A^{\prime}E$ state.

Finally, in the case where we wish to implement such a channel using some additional resources, such as a shared key or private randomness, we may extend $\rho$ with additional registers and add conditions as necessary to specify the form of the resources.  In keeping with the spirit of the definition, we will only consider resources where the insider has access to any stored data, including shared private keys.

It is interesting to note that, compared with the usual definition of a private channel, the only difference is that the private key is allowed to be correlated with the message.

\subsubsection{A naive attempt}
To construct an insider-proof private channel is non-trivial.  Consider a naive application of the one-time pad where Alice simply XORs her message with a shared private key and sends it to Bob over a public channel.  This  can be decrypted by Bob and uses a key unknown to Eve.  Let $d$ be some of Alice's data, and $k$ be the key.  The insider can choose to send the message $m = d \oplus k$ to Bob.  Alice then encrypts the message as $m \oplus k = d$, which is broadcast to Bob and Eve.  Hence Eve can trivially recover $d$.  Note that Eve could not build her final state by herself since it is correlated with the register $D$, to which she has no access.

Clearly if Alice uses an encryption scheme which depends only on her stored data (which the insider has access to), the insider can, at least information theoretically, reverse the encryption and choose a message which results in a ciphertext that directly reveals data to Eve.  Hence we introduce some randomness in the form of a true random number generator.  Then Alice encrypts her message by using some function which depends on a random string which is chosen \emph{after} the message is chosen.  In this case the insider cannot predict what the encryption function will be.  It is thus our task to find a suitable function.

\subsection{The channel}
\label{sec:implch}
In order to achieve our goal, we must use some additional resources in the form of a shared private key and a true random number generator on Alice's side.  The shared key is in register $K$ while Alice's private random string is held in $R$.  The initial state must satisfy $\rho_{KRBCDE} = U^{(n)}_{K} \otimes U_{R}^{(n)} \otimes\rho_{BCDE}$ and $\rho_{RABCDE} = U^{(n)}_{R} \otimes \rho_{ABCDE}$, where $U^{(n)}$ is the completely mixed state on $n$ qubits.  Note that the message can be correlated with the shared key, but the rest of the state cannot.  As well, the private randomness is uncorrelated with all other registers.  Our protocol is summarized as follows.

\begin{protocol}\label{protocol}
Input for Alice: strings $a,k$.  Input for Bob: string $k$.
\begin{enumerate}
	\item Alice chooses a string $r$ uniformly at random.
	\item Alice calculates $c  = a \oplus \left[(k \cdot r) \mod 2^{\ell}\right]$ and discards $k$.
	\item Alice broadcasts $(c,r)$ and then discards them.
	\item Bob reconstructs $a = c \oplus \left[(k \cdot r) \mod 2^{\ell}\right]$ and then discards $k$, $r$ and $c$.
\end{enumerate}
\end{protocol}

In order to prove that protocol~\ref{protocol} produces an $\epsilon$-insider-proof private channel we first introduce 2-universal hash functions.

\subsection{Security of the channel}
\label{sec:secch}

We first sketch the proof, then provide the technical details.  The insider $A^{\prime}$ chooses some message $A$ with full knowledge of the shared key $K$.  Hence the message can be correlated with $K$.  However, we will use a $K$ of length more than twice that of $|A|$ ($|K| > 2|A|$) so that there are still $> |A|$ bits of randomness in $K$, even conditioned on $A$.  Now when we produce the encryption key $K^{\prime}$ by combining $K$ with $R$, we produce a $K^{\prime}$ of length $|A|$.  The leftover hashing lemma (stated below) then says that $K^{\prime}$ is almost completely random, even conditioned on $A$ and $R$.  The ciphertext is then also completely random, even conditioned on $A$ and $R$ and Eve will not be able to figure out anything about $A$ from $R$ and the ciphertext.

\subsubsection{2-universal hash functions}\label{sec:hash}
2-universal hash functions are in fact families of functions which, given a random seed, produce a very uniform output.

\begin{definition}
A \emph{2-universal} family of functions $\mathcal{F}$ is a family of functions $f: \mathcal{X} \rightarrow \mathcal{Y}$ such that, when $f$ is drawn uniformly at random from $\mathcal{F}$, for every $x_{1}, x_{2} \in \mathcal{X}$
\begin{equation}
P(f(x_{1}) = f(x_{2})) = \frac{1}{|\mathcal{Y}|}
\end{equation}
¥
\end{definition}

Protocol \ref{protocol} uses the following 2-universal family of hash functions introduced in \cite{CW79}.

\begin{lemma}\label{lemma:2universal}
The family of functions given by $f_{r}(k) = k \cdot r \mod 2^{\ell}$  is 2-unversal.
\end{lemma}

\begin{proof}

Let $x_{1} \neq x_{2}$ be given.  We wish to count the $r$ for which
\begin{equation}
(r \cdot x_{1}) = (r \cdot x_{2}) \mod 2^{\ell}.
\end{equation}
Taking the expression $\mod 2^{\ell}$, i.e. taking the $\ell$ least significant bits of the string, can be expressed as taking the expression $\mod b$ for some element\footnote{In particular, the element $x^{\ell}$ in the usual polynomial representation.} $b$ in $\text{GF}(2^{n})$.  Hence we can rewrite this as
\begin{equation}
r \cdot ( x_{1} \oplus  x_{2}) = 0  \mod b.
\end{equation}
Since the multiplication is over a field, $r = 0 \mod b$ and there is one solution for every member of the equivalence class of $0$, of which there are $2^{n-\ell}$ members.  Hence the fraction of strings $r$ that are solutions is $2^{n-\ell} / 2^{n} = 2^{\ell}$ and the family of functions is 2-universal.
\end{proof}

Note that the family is symmetric in the roles of $r$ and $k$, so we can use $k$ as the seed instead of $r$ and the family is still 2-universal.  The distinction becomes important in the following lemma, which gives a useful approximation of how uniform the output of the hash function is.

\begin{lemma}[Quantum leftover hashing lemma]~\cite{TSSR11}
\label{lemma:leftoverhashing}
Let $X$ and $E$ be random variables.  Let $\mathcal{F}$ be a family of 2-universal hash functions, indexed by a seed $R$ such that $f_{R} \in \mathcal{F}$, that take an input $X\in\{0,1\}^n$, and output $Z\in\{0,1\}^\ell$.  Then averaged over $f_R$, the distribution on $Z$ has the property:
\begin{equation}
\Delta(Z|ER) \leq \epsilon' + \frac{1}{2} \sqrt{2^{\ell - H_{\min}^{\epsilon'}(X|E)}} \, ,
\end{equation}
where the distance from uniform, $\Delta$, is given by
\begin{equation}
\Delta(A|B)_{\rho} = \min_{\sigma_B} \frac{1}{2} \left\lVert \rho_{AB} - \omega_A \otimes \sigma_B  \right\rVert_1 \, .
\end{equation}
\end{lemma}
%

\subsubsection{Proof of security}

\begin{theorem}
Let $\ell$ and $n > 2\ell$ be given. Then protocol~\ref{protocol} implements an $\epsilon$-insider-proof secure channel where 
\begin{equation}
\epsilon = \sqrt{2^{2\ell - n}}
\end{equation}

\label{thm1}
\end{theorem}

\begin{proof}

We begin by reducing to an equivalent protocol by noting that, so long as Bob completes the protocol before interacting with outside parties, his operations commute with Eve's.  Hence we may assume that Eve receives her copy of $(c,r)$ after Bob has completed the protocol.  This solves certain notational problems where we need to trace out registers in the proper sequence in order to obtain valid bounds.  Also, in this version of the protocol, we make explicit the movement of registers between different parties.

\begin{protocol}\label{protocol2}
Input for Alice: Registers $A$ and $K$.  Input for Bob: Register $J$.
\begin{enumerate}
	\item Alice uses her random number generator to initialize $R$ with a uniformly random string.
	\item Alice calculates $K^{\prime} = (K \cdot R) \mod 2^{\ell}$, then discards $K$ and sends $R$ to Bob.
	\item Bob calculates $J^{\prime} = (J \cdot R) \mod 2^{\ell}$, then discards $J$.
	\item Alice calculates $C = K^{\prime} \oplus A$ and then discards $K^{\prime}$ and sends $C$ to Bob.
	\item Bob calculates $B = C \oplus J^{\prime}$ and then discards $J$.
	\item Bob passes $C$ and $R$ to Eve.
\end{enumerate}
\end{protocol}

Here $K$ contains the private shared key, of which $J$ is Bob's copy.  

Next we make a further reduction.  At the end of step 5, Bob's state consists solely of $B$, which is a copy of $A$.  Hence we can instead simply apply $\Phi^{IPC}_{AB}$ at the end of the protocol and remove all of Bob's operations, as well as $J$.  Then Alice can simply send $C$ and $R$ directly to Eve.  Hence we arrive at the following protocol

\begin{protocol}\label{protocol3}
Input for Alice: Registers $A$ and $K$.
\begin{enumerate}
	\item Alice uses her random number generator to initialize $R$
	\item Alice calculates $K^{\prime} = (K \cdot R) \mod 2^{\ell}$ and discards $K$.
	\item Alice calculates $C = A \oplus K^{\prime}$ and discards $K^{\prime}$.
	\item Alice sends $C$ and $R$ to Eve.
	\item Alice and Bob apply $\Phi^{IPC}_{AB}$
\end{enumerate}

\end{protocol}

Now we proceed with the security proof.  Let $\ell = |A|$ and $n = |K| = |R|$.  For notational convenience we assume that $C$ is created in step 3, and $B$ is created in step 5, so we need not keep track of them beforehand.  Let the quantum state just after step $t$ be $\rho^{(t)}$ and the $\epsilon$-smooth min-entropy be $H_{\min}^{\epsilon}$.  We suppose for the moment that the key in register $K$ is perfectly independent from Eve.

After step 1, since $K$ is secret from Eve, $H_{min}^{0}(K | DE)_{\rho^{(1)}} = n$.  By the chain rule for smooth entropies \cite{rennerthesis}, we also have 
\begin{equation}
H_{min}^{0}(K|ADE)_{\rho^{(1)}} \geq n - \ell.
\end{equation}
In step 2 we apply the 2-universal hash given in lemma~\ref{lemma:2universal}, tracing out $K$ and producing encryption key $K^{\prime}$ of length $\ell$.  Using the leftover hashing lemma we find
\begin{equation}
\Delta(K^{\prime}|ADER)_{\rho^{(2)}} \leq \frac{1}{2} \sqrt{2^{2\ell - n}} = \epsilon_{hash}
\end{equation}
and hence there exists a $\sigma_{ADER}$ such that
\begin{equation}
\norm{\rho^{(2)}_{K^{\prime}ADER} - U_{K^{\prime}} \otimes \sigma_{ADER}}_{1} \leq \epsilon_{hash}.
\end{equation}

Since for $U_{K^{\prime}} \otimes \sigma_{ADER}$ $A$ is independent of $K^{\prime}$, we can XOR them together in step 3 to obtain the ciphertext $C$ which is again independent.  We trace out $K^{\prime}$ and then $C$ and $R$ are sent to Eve in step 4.  We find 
\begin{equation}
\norm{\rho^{(4)}_{CADER} - U_{C} \otimes \sigma_{ADER}}_{1} \leq \epsilon_{hash}.
\label{eq:rho4}
\end{equation}

Next we want to approximate $\sigma_{ADER}$.  Tracing out the $C$ register, the above inequality becomes $\norm{\rho^{(4)}_{ADER} - \sigma_{ADER}} \leq \epsilon_{hash}$.  Since $\rho^{(4)}_{ADER} =  \rho^{(1)}_{ADER} =  \rho_{ADE} \otimes U_{R}$ we then obtain 
\begin{equation}
\norm{U_{C} \otimes \sigma_{ADER} - U_{C} \otimes \rho^{(0)}_{ADE}Ê\otimes U_{R}} \leq \epsilon_{hash} \ .
\label{eq:unif}
\end{equation}
Now $U_{C} \otimes \rho_{ADE} \otimes U_{R}$ is a state that Eve can create by herself by operating only on her registers by simply appending $C$ and $R$ distributed uniformly.  Let us call this operation $\Psi$.  Using the triangle inequality to combine~(\ref{eq:rho4}) and~(\ref{eq:unif}),
\begin{equation}
\norm{\rho^{(4)}_{ACDER} - I_{AD} \otimes \Psi_{CER}(\rho^{(0)}) }_{1} \leq 2\epsilon_{hash} = \epsilon \ .
\end{equation}
We now introduce register $B$ and after step 5, this becomes
\begin{equation}
\norm{\rho^{(5)}_{ABCDER} - \Phi^{IPC}_{AB} \otimes I_{D} \otimes \Psi_{CER}(\rho^{(0)}) }_{1} \leq \epsilon \ .
\label{eq:secbnd}
\end{equation}
Hence the protocol implements a $2\epsilon_{hash}$-insider-proof private channel.   
\end{proof}

\section{Application to DIQKD}
We now consider the application of the $\epsilon$-insider proof channel to DIQKD in the context of reused devices with memory.

\subsection{The model}
\label{sec:model}

Alice and Bob share some private randomness and would like to grow more key from it using a shared quantum state.  However, they do not trust their measuring devices or the state; in fact, they assume that Eve has built the devices and distributes the quantum state.  Let us assume that it is possible for them to complete a device-independent quantum key distribution (DIQKD) protocol securely in this setting.  There is some recent work that supports this assumption~\cite{BCK12-2,RUV12-2,RUV12-1}.  They successfully grow some new key on which Eve's knowledge is bounded to be less than $\epsilon$, quantified using standard trace distance metrics~\cite{rennerthesis}.  After this, they would like to \emph{reuse their devices} to grow more key in another round, but the malicious devices are allowed to have memories.  As well, all shared randomness used in the protocol will be taken from the previously generated keys, and hence is also shared with the devices.%
\footnote{
At the very least, the devices can know the raw keys from previous rounds, and hence are strongly correlated with the final keys.}
We would like to know whether Alice and Bob can grow new key in this situation.  

We make the standard assumptions of DIQKD.  We are working in the limit of long keys for each run of the protocol.  We assume that the untrusted devices can be isolated within Alice and Bob's laboratories, such that they can receive arbitrary quantum signals from Eve, but can signal only to Alice and Bob and not directly to Eve.  We also assume that Alice and Bob can both generate trusted randomness locally.  Additionally, we assume Alice and Bob can perform classical processing privately from the untrusted measuring devices in their labs.  

This model was first introduced in~\cite{BCK12}, where the authors argue that in standard protocols Alice and Bob cannot grow further key using the same devices.  Particularly, they highlight the issue of whether the protocols are composable.  We show how to modify standard DIQKD protocols to eliminate side channels related to Alice and Bob's public discussion and show that they can still grow new secret key.  We comment on the issue of composability in section~\ref{sec:aborts}.

\subsection{The protocol}
\label{sec:protocol}
The modifications we propose are restricted to the classical post-processing portions of the protocol.  The goal of the changes is to prevent the device from having a communication channel back to Eve within the protocol itself.  To this end, we make use of an $\epsilon$-insider-proof channel to send all information between Alice and Bob that the untrusted devices may have influenced.  (We assume no other side channels.)

Our modification applies to DIQKD protocols with standard classical post-processing~\cite{rennerthesis}.  Importantly, with standard post-processing the only information communicated between Alice and Bob which depends on the quantum devices are the parameter estimation data, the error correction data, and the abort flag. 

\begin{enumerate}
\item  Eve distributes an entangled state $\rho_{ABE}$ to the devices in Alice and Bob's labs.  Alice and Bob supply random (and independent) lists of basis choices to the devices for the series of measurements and the devices output the results.  
\item  Alice announces her basis selections publicly to Bob.  Where they have chosen the same basis, the measurement result bit should be correlated for Alice and Bob and can become part of the key.  When they have chosen different bases, they can check for CHSH violation or perform other parameter estimations.
\item  Alice must send to Bob a subset of her outcomes of size $\ell$.  To do this, they use protocol~\ref{protocol} to implement an $\epsilon$-insider-proof private channel, which must not leak information about previously grown keys (or other private data), $d$.  The message string $a = a(k,d)$ is passed from Alice to Bob, encrypted.  To do this, she generates a random string $r$ ($|r| = n)$ and chooses a string $k$ ($|k| = n$) from her store of previously generated keys.  She uses the type of 2-universal hash function introduced in lemma~\ref{lemma:2universal} to create ciphertext $c  = a \oplus (k \cdot r) \mod 2^{\ell}$.  She sends this to Bob along with $r$.  Bob uses $r$ and $k$ to recover $a$.  
\item Bob performs parameter estimation.  He sends a similarly encrypted message to Alice containing a flag bit indicating abort or not, and if not, a second encrypted message containing the detected bit error rate $Q$, the observed parameters, and an appropriate error correction function, along with his parity check bits.  Bob pads this communication with randomness, so it is always of fixed length.  If they instead will abort, Bob sends the abort flag and a random message instead of the error correction information. 
\item  Alice uses the information to correct her string to Bob's. 
\item  Using a publicly chosen hash function they perform privacy amplification to reduce Eve's knowledge of the final key below a chosen bound.  They discard the session encryption key $k$ used in the protocol.
\end{enumerate}

We now show the security of this protocol.

\subsection{Security of QKD using an insider-proof channel}
\label{sec:qkdsec}

In order to complete a QKD protocol Alice and Bob will require a series of communication channels back and forth which they have authenticated.  When the devices in Alice and Bob's labs may have some sensitive information in their memories, then some of these channels must be private channels, in order to show security.

\begin{figure}[h!]
\label{fig:indiqkd}
\includegraphics[scale=0.65]{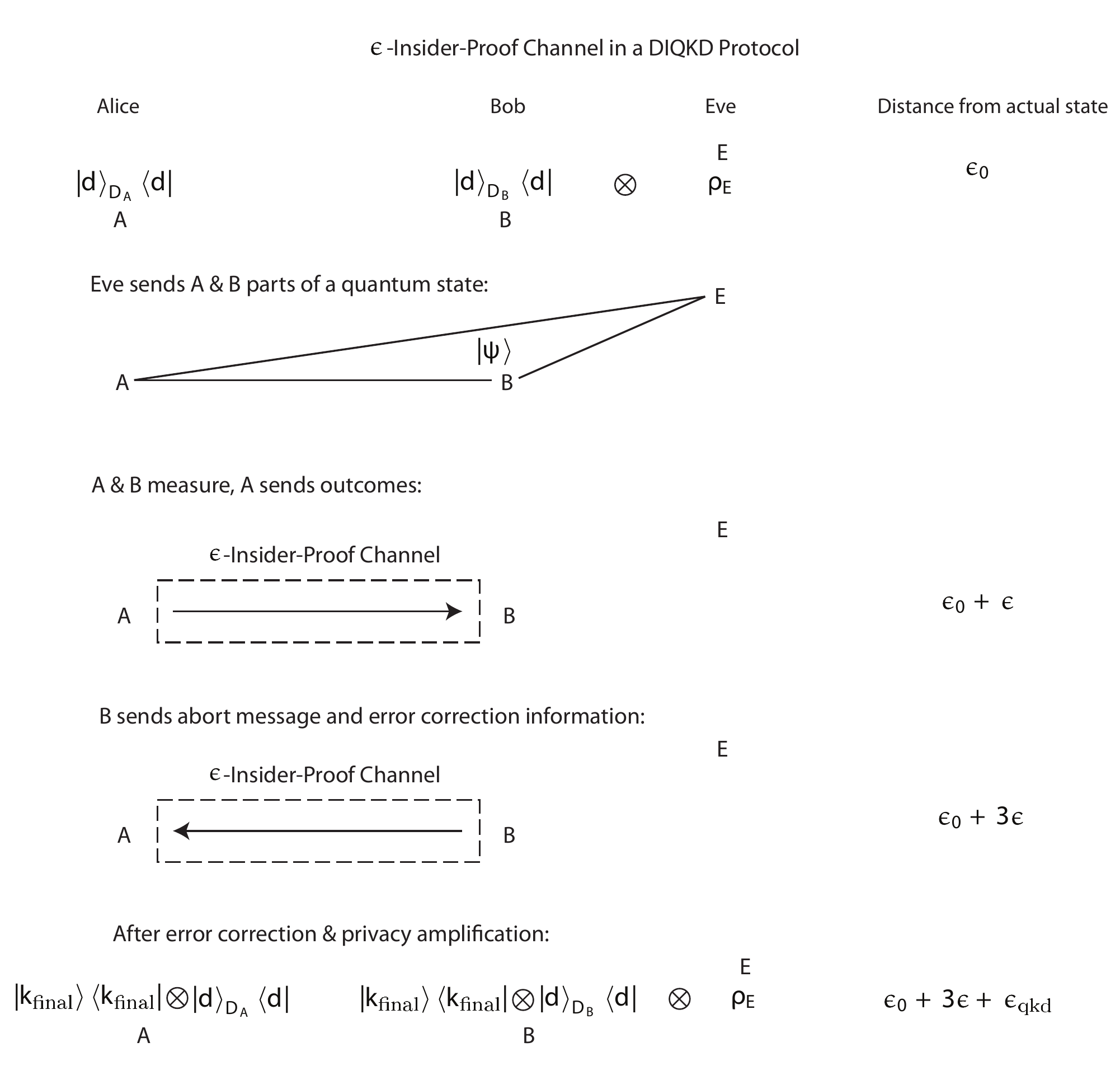}
\caption{Use of the insider-proof channel in a device independent quantum key distribution protocol.}
\end{figure}

Again, let the quantum state just after step $t$ be $\rho^{(t)}$.  At first, let us analyze the protocol assuming we start with a perfect key so that $p^{(0)} = U_{K}^{(3n)} \otimes \sigma_{D} \otimes \tau_{E}$.  After step 1, Alice and Bob share with Eve the state $\rho_{ABE}$.  They pick measurements and get outcomes in registers $A'$ and $B'$, so that their shared state becomes $\rho^{(1)}$, where
\begin{eqnarray}
\rho^{(1)} &=& U_{K}^{(3n)} \otimes \sigma_D \ \otimes \sum_{o_{m_{A}},o'_{m_{B}}} p(o_{m_{A}},o'_{m_{B}}) \ket{o_{m_{A}}}_{A'}\bra{o_{m_{A}}}     \nonumber  \\
& & \qquad \qquad \qquad \qquad \qquad \otimes \ket{o'_{m_{B}}}_{B'}\bra{o'_{m_{B}}} \otimes \rho^{(o_{m_{A}},o'_{m_{B}})}_{E} \ .
\end{eqnarray}

Now in step 2, Alice uses a public channel to send Bob her measurement choices $m_A$ and Bob can also use a public channel to send Alice his choices $m_B$.  Alice will prepare a private message for Bob that includes a subset $a$ of her outcomes $o_{m_{A}}$.  She then implements (in step 3) an insider-proof quantum channel to Bob, according to protocol~\ref{protocol3}.  

We can alter $\Phi_{ABC}$ to take the string in register $K$ as part of the input state rather than a parameter that defines $\Phi_{ABC}$.  In all other respects, the channel is unchanged.  Let the new channel be $\Phi'_{ABC}$.  Then from definition~\ref{def:ippc},
\begin{equation}
\norm{\Tr_K \Phi_{AB}^{IPC} \otimes I_{D} \otimes \Psi_{CE}(U_K\otimes\rho) - 
\Tr_K \Phi'_{ABC} \otimes I_{DE}(U_K\otimes\rho)}_{1} \leq \epsilon \ .
\label{eq:chbnd}
\end{equation}

Let the register $A$ contain the subset of outcomes (so $a$ is a function of $o_{m_{A}}$.  Then after an ideal private channel the state will be $\rho^{(3)}$ such that:
\begin{equation}
\norm{ \rho^{(3)} - \xi^{(3)}}_1 \leq \epsilon \ .
\end{equation}
where
\begin{eqnarray}
\xi^{(3)} &:=& U_K^{(2n)} \otimes \sigma_D \otimes \sum_{o_{m_{A}},o'_{m_{B}}} p(o_{m_{A}},o'_{m_{B}}) \ket{a}_A\bra{a} \otimes \ket{o_{m_{A}}}_{A'}\bra{o_{m_{A}}}  \nonumber   \\
& & \qquad \qquad \otimes \ket{a}_B\bra{a} \otimes \ket{o'_{m_{B}}}_{B'}\bra{o'_{m_{B}}} \otimes \rho^{(o_{m_{A}},o'_{m_{B}})}_{E}\otimes \II_{C,R} \ .
\end{eqnarray}
Bob will also have to reply in step 4, again using an insider-proof channel twice.  First he sends a one-bit message about whether to abort and second he sends the error correction information.  For an ideal private channel:
\begin{equation}
\norm{\rho^{(4)} - \xi^{(4)}}_1 \leq 3\epsilon \ .
\end{equation}
where
\begin{eqnarray}
\xi^{(4)} &:=& \sigma_D \otimes \sum_{o_{m_{A}},o'_{m_{B}}} p(o_{m_{A}},o'_{m_{B}}) \ket{b}_A\bra{b} \otimes \ket{o_{m_{A}}}_{A'}\bra{o_{m_{A}}}   \nonumber \\
& & \qquad \ \ \otimes \ket{b}_B\bra{b} \otimes \ket{o'_{m_{B}}}_{B'}\bra{o'_{m_{B}}} \otimes \rho^{(o_{m_{A}},o'_{m_{B}})}_{E} \otimes (\II_{C,R})^{\otimes 3}  \, .
\end{eqnarray}
At this point, they arrive at identical raw keys with probability $1-\epsilon_{EC}$, where Alice and Bob can choose $\epsilon_{EC}$ arbitrarily small.  Then, 
\begin{equation}
\norm{ \rho^{(5)} - \xi^{(5)}}_1 \leq  3\epsilon + \epsilon_{\text{EC}}  + \epsilon_{\text{PE}} \ ,
\end{equation}
defining
\begin{equation}
\xi^{(5)} := \sigma_D \otimes \sum_{o_{m_{A}},o'_{m_{B}}} p(o_{m_{A}},o'_{m_{B}}) \ket{k_{\text{raw}}}
_{A'}\bra{k_{\text{raw}}} \otimes \ket{k_{\text{raw}}}_{B'}\bra{k_{\text{raw}}} \otimes \rho^{(o_{m_{A}},o'_{m_{B}})}_{E} \ ,
\end{equation}
where we dropped the $C$ and $R$ registers for convenience, and $k_{\text{raw}}$ still depends on $o_{m_{A}}$ and $o'_{m_{B}}$.
They then implement a privacy amplification hash in step 6 and let us define $\epsilon_{\text{qkd}} = \epsilon_{\text{EC}} + \epsilon_{\text{PE}} + \epsilon_{\text{PA}}$.  So now we are left with a state $\rho^{(6)}_{A'B'CDER}$ such that:
\begin{equation}
\norm{\rho^{(6)}_{A'B'CDER} - \mathcal{U}_{A'B'} \otimes \sigma_{D} \otimes \tau'_{CER} }_{1} \leq  3\epsilon+\epsilon_{\text{qkd}} \ .
\label{eq:second}
\end{equation}
where $\mathcal{U}_{A'B'}$ is the normalized uniform distribution over all strings of a given length and we have followed the standard analysis (see \emph{eg.}~\cite{rennerthesis}) for the overheads of a single round of QKD.  We can write this instead as
\begin{equation}
\norm{\Phi_{\text{prot}}(U_{K}\otimes \sigma_{D} \otimes \tau_{E}) - \Phi_{\text{ideal}}(U_{K}\otimes \sigma_{D} \otimes \tau_{E}) }_{1} \leq  3\epsilon+\epsilon_{\text{qkd}} \ .
\label{eq:second}
\end{equation}
where $\Phi_{\text{prot}}$ is the action of the entire modified QKD protocol and $\Phi_{\text{ideal}}$ is an ideal protocol that shares key between Alice and Bob while leaking nothing to Eve.  

Now let us relax the assumption of a perfect key.  Instead, assume that Alice and Bob have already successfully grown some key using a DIQKD protocol, secure against malicious devices with memory.  Before step 1, we assume that Eve has bounded correlations with these keys:
\begin{equation}
\norm{\rho_{KDE} - U_{K} \otimes \sigma_{D} \otimes \tau_{E} }_{1} \leq \epsilon_0 \ .
\label{eq:init}
\end{equation}
We can apply $\Phi_{\text{prot}}$ to both states in the above bound.  Then using the data processing inequality, we have
\begin{equation}
\norm{\Phi_{\text{prot}}(\rho_{KDE}) - \Phi_{\text{prot}}(U_{K} \otimes \sigma_{D} \otimes \tau_{E}) }_{1} \leq \epsilon_0 \ .
\label{eq:dp}
\end{equation}
We can use the triangle inequality on equations~(\ref{eq:second}) and~(\ref{eq:dp}), to finally obtain
\begin{equation}
\norm{\Phi_{\text{prot}}(\rho_{KDE}) - \Phi_{\text{ideal}}(U_{K} \otimes \sigma_{D} \otimes \tau_{E}) }_{1} \leq \epsilon_0 + 3\epsilon+\epsilon_{\text{qkd}} \ .
\label{eq:dp}
\end{equation}

Now, let us back up a minute and consider what happens if Alice and Bob need to abort in step 4.  Implementing the insider-proof channel uses up their store of private key.  Asymptotically, the largest amount of key will be used to send the error correction information.  However, if they abort, there is no need to send this.  By using separate applications of the channel, after sending the signal to abort, Bob is free to not use the insider-proof channel and instead send a random string.  This is fine, since referring to protocol~\ref{protocol} the contents of $R$ are uniformly random, and, looking at equation~(\ref{eq:rho4}), the contents of $C$ cannot be distinguished from a uniform string by the adversary, except with probability $\epsilon$.  Therefore, in the case of an abort, the largest share of the cost of establishing a insider-proof channel can be avoided by breaking up Bob's messages in this way.

\subsection{Composing rounds of the new protocol}
\label{sec:bounds}

In the previous section, we saw that reusing untrusted devices in a new round of QKD using the new protocol caused an increase in the security parameter of the new and old keys by $3\epsilon + \epsilon_{\text{qkd}}$.  For comparison, if the devices were trusted, and the original DIQKD protocol was used, this parameter would only have grown by $\epsilon_{\text{qkd}}$.

Then composing $s$ rounds of successful key growth together, 
\begin{equation}
\norm{\Phi_{\text{prot}}^{\circ s}(U_{K}\otimes \sigma_{D} \otimes \tau_{E}) - \Phi_{\text{ideal}}^{\circ s}(U_{K}\otimes \sigma_{D} \otimes \tau_{E}) }_{1} \leq  3\epsilon+\epsilon_{\text{qkd}} \ .
\label{eq:newsecond}
\end{equation}
where $\Phi^{\circ s}$ means the channel $\Phi$ applied $s$ times.  Again using the data processing inequality for $s$ applications of $\Phi_{\text{prot}}$ on equation~(\ref{eq:init}) and then the triangle inequality with equation~(\ref{eq:newsecond}) gives
\begin{equation}
\norm{ \Phi_{\text{prot}}^{\circ s}(\rho_{KDE}) - \Phi_{\text{ideal}}^{\circ s}(U_{K}\otimes \sigma_{D} \otimes \tau_{E})  }_1  \leq \epsilon_0 + 3s\epsilon + s\epsilon_{\text{qkd}}  \, .
\end{equation}
This shows that each additional round can add at most $3\epsilon + \epsilon_{\text{qkd}}$ to Eve's information on the previously grown keys.  

Notice that if an abort occurs in round $i$, the new key is not obtained for that round, so the length of the final key string will depend on the number of aborts as well as the error rates.  However, Alice and Bob still sent two encrypted messages to each other in an aborted round, in order to learn that their error rate was above threshold.  Therefore, they still must add $3\epsilon$ for that round, though not $\epsilon_{\text{qkd}}$.  This means that the security parameter will grow even on aborted rounds.  

In practice, Alice and Bob should choose a maximum tolerated security loss of all of their keys $\epsilon_{\text{sec}}$.  This will determine the number of rounds they would be able to grow key in.  They should agree to this number of rounds when they begin to use their devices, then stop using and securely destroy the devices after that many rounds.  They do not wish to leak information to Eve about the number of rounds that have aborted.  (See section~\ref{sec:aborts} for further discussion.)

Note that this growth of the security parameter with the number of rounds is also seen in the standard trusted-device QKD models when some of the grown key is used for authentication in subsequent rounds.

\subsection{Asymptotic secret key rate}
\label{sec:accounting}

The application key rates achievable with this protocol modification will depend on the key rate of the underlying DIQKD protocol used, and $n$ the number of bits of the generated key that need to be used as the session keys for Alice and Bob's encrypted messages in the next round, and therefore cannot be used in other applications.  

Since we do not know the details of which DIQKD protocol can be used when the devices have memories, we remain agnostic about the exact rate, however, we can assume it would take a form:
\begin{equation}
r \geq f(S_{\text{obs}})  - H(A|B)
\label{eq:rate}
\end{equation}
for some function $f$ with $S_{\text{obs}}$ an observed parameter (\emph{eg.} a Bell-inequality violation) which is what is achieved by current protocols against memoryless devices~\cite{MPA11,HR10}.

In this new protocol, we do not need to remove the amount of communication $H(A|B)$ required for error correction, since this is encrypted.  However, we will remove the amount of key required to encrypt the next round's communication.  We now consider how much key this requires.  From theorem~\ref{thm1}, we have:
\begin{equation}
\epsilon = \sqrt{2^{2\ell - n}} = \frac{1}{2^{(n-2\ell)/2}} \, ,
\end{equation}

\noindent Then $n-2\ell = O(-\log\epsilon)$, so for a constant security parameter $\epsilon$, the key length, $n$, needs only exceed twice the message length, $2\ell$, by a constant number of bits.

Now we must determine how large the total amount of encrypted information sent between Alice and Bob must be asymptotically.  Suppose the sifted key length in one round is $N$.  The parameter estimation message from Alice to Bob must contain the bit values of an $O(\log N)$-size subset of this string in order to achieve an estimation error approaching zero.  As $N\rightarrow \infty$ the fraction of signals this represents goes to zero.  Bob must send to Alice his error correction function results, the size of which will depend on the error rate.  The amount of communication required will be $H(A|B)+f(\epsilon_{\text{EC}})$ bits, where $f(\epsilon_{\text{EC}})$ is a function of the security parameter for the error correction that does not depend on $N$, so that as $N\rightarrow \infty$ it also is negligible.  Finally, Bob's abort flag requires a constant sized key.  

In total, asymptotically, the amount of key needed to implement the insider-proof channels in the protocol depends only on the size of the error correction information to be shared.  Since we have $n\geq 2\ell +c$ where $c$ is a constant, the amount of key required is just twice the error rate: $2H(A|B)$.

Then we can see how the asymptotic key rate will change as compared with the original version of the protocol,   
\begin{equation}
r \geq f(S_{\text{obs}})  - 2H(A|B) \, .
\end{equation}
Notice that asymptotically the key rate does not fall as aborts occur, since in an abort, Bob will send the encoded abort flag, but will not encode the $H(A|B)$ bits of error correction information and rather save his key by sending a string output by his random number generator instead.  In the finite key regime however, it is clear that aborts will reduce the amount of generated key that can be used in other applications.

\subsection{Aborts}
\label{sec:aborts}

It may happen that on some rounds Alice and Bob must abort the protocol.  However, since the devices that Alice and Bob use can cause an abort even on a ``good'' state $\rho_{A'B'E}$, they can use this as a pretext to signal to Eve, as was observed in~\cite{BCK12}.  Therefore, Alice and Bob must hide aborts when they occur.  As explained in section~\ref{sec:protocol}, they can do this since they have encrypted the parameter estimation bits and will also encrypt Bob's signal as to whether or not to abort.  If they abort, they pretend to continue the protocol, but instead of exchanging encrypted information to perform error correction, they send random strings.  In this round they do not gain any additional key, but also Eve does not learn that they aborted.  

Another concern is that it is possible for the boxes to conduct a denial-of-service attack until Alice and Bob run out of key.  If this should occur before the number of rounds that Alice and Bob had agreed to use the devices for, this would also constitute a signal to Eve.  They must hide this also, so should it occur, Alice and Bob should simulate the remaining rounds of key growth (sending each other random strings) and then destroy the adversarial boxes securely.  This is not a foolproof solution however, since in the meantime Alice and Bob may need to communicate privately.  Thus at some point they will be forced to re-key and there is no reason to assume Eve will not notice this.  Therefore, it is conceivable that she may gain some information from the fact that this has happened and it seems there is no way to completely avoid that, though Alice and Bob could keep a piece of their initial authentication key from before the first round against this eventuality.  (This is similar to the case in trusted-device QKD when Eve executes repeated denial-of-service attacks on Alice and Bob until they run out of key.)

It appears that in this model we cannot think about each run of the device independent protocol as a stand-alone element in a universal composability scheme, in which it is public information how much key they have at any given time.  Alice and Bob certainly do not want to output on each round whether they succeeded or failed in obtaining key.  This may lead to additional considerations.  For example, the adversary may expect Alice and Bob to send a one-time-pad encoded message at a particular time during the multi-round life of the devices when they do not have key available to devote to the purpose.  If this occurs they can still avoid leaking information to the adversary by sending a random string of the appropriate length instead.  (However, this does not accomplish the communication task Alice and Bob presumably wished to accomplish.)  Note that in this case, Alice and Bob have to consider their quantum key distribution in the wider setting in which it is employed to avoid leaking information.  Nevertheless, when key is generated in the DIQKD scheme, the resulting key is secure under the trace distance definition given in~\cite{rennerthesis}.

\section{Conclusions}

We have introduced the concept of an insider-proof channel.  We hope that it will have applications, particularly in device-independent schemes where untrusted devices can be assumed not to have direct communication to the adversary, but may be malicious.  We construct an explicit example of such a channel that will allow trusted parties to communicate, even about information that the untrusted devices may have generated.  We also show how this can be used to reuse untrusted devices for many rounds of QKD.

The model of DIQKD assumed here gives a lot of power to the eavesdropper, since Eve is allowed to prepare Alice and Bob's measuring devices.  It is more restrictive to Alice and Bob than other models currently used to describe untrusted device scenarios, where their devices may have manufacturing flaws, but are assumed not to be outright malicious.  Those models more realistically represent most cryptographic scenarios today, wherein perhaps a user does not understand the cryptography implemented by his web browser, but he downloaded an authenticated copy from a legitimate business.  The business may not have correctly implemented the security, and this is what DIQKD would try to protect against, but it also does not benefit from gaining a reputation for selling users' credit card information to Eve.

However, this less-trusting model is interesting, first, because it provides bounds for what is possible in other more-trusting DI scenarios, and second, because despite its restrictions, QKD can still be performed without much loss of performance.  We have introduced a small modification to a DIQKD protocol that allows untrusted and malicious devices to be used in repeated round of secure key growth.  It is interesting to note that the only part of the protocol that requires modification is the classical post-processing.  This suggests that perhaps existing QKD protocols could be adapted to other new models readily, simply by considering this portion carefully.

There remain some open questions.  Are there other applications for insider-proof channels?  More specfically, are there other contexts where messages may be chosen maliciously and with knowledge of private data? 
  It also may be possible to improve the bounds presented here in order to get a higher asymptotic key rate.   It would also be nice to fit this type of protocol into a composability framework, although it is not clear how to do that in existing frameworks.  Additionally, there may be other modifications that could be made to existing protocols that accomplish this same task more efficiently.

{\bf Acknowledgements} This work is funded by the Centre for Quantum Technologies, which is funded by the Singapore Ministry of Education and the Singapore National Research Foundation, and by the University of Otago through a University of Otago Research Grant and the Performance Based Research Fund.  We thank Marco Tomamichel for a helpful discussion and Roger Colbeck for his comments about composability. 

\bibliographystyle{plain}
\bibliography{QKDBib13}

\appendix

\end{document}